\documentclass[a4paper,11pt]{article}
\usepackage[margin=1in]{geometry}

\usepackage[english]{babel}
\usepackage{lmodern}
\usepackage[T1]{fontenc}
\usepackage{amsmath,amsfonts,amsthm}
\usepackage{url}
\usepackage{hyperref}
\usepackage{graphicx}
\usepackage[utf8]{inputenc}
\usepackage{enumerate}
\usepackage{paralist}
\usepackage{comment}
\usepackage[usenames]{color}
\usepackage{authblk}

\RequirePackage{lineno}

\newcommand{\lastcorrections}%
{{
 \begin{sloppypar}
    \baselineskip -0.2in
    \tiny\bf\noindent
last corrections:\\
\end{sloppypar}
}}

\newcommand{\margincomment}[1]%
    {{%
      \marginpar{{\scriptsize\begin{minipage}{0.75in}
                       \begin{flushleft}
                          {#1}
                       \end{flushleft}
                       \end{minipage}
                }}
    }}
\renewcommand{\margincomment}[1]{}

\newcommand{\lref}[2][]{\hyperref[#2]{#1~\ref*{#2}}}

\newcommand{\myparagraph}[1]{{\smallskip\noindent{\bf #1}}}
\newcommand{\mymedparagraph}[1]{{\medskip\noindent{\bf #1}}}

\newcommand{\etal}{{\em et al.}}

\newcommand{\barS}{{\bar S}}

\newcommand{\calA}{{\cal A}}

\newcommand{\calG}{{\cal G}}

\newcommand{\E}{\mathbf{E}}

\newcommand{\braced}[1]{{ \left\{ #1 \right\} }}

\renewcommand{\Pr}{\textrm{Pr}}

\newcommand{\half}{{\mbox{$1\over 2$}}}

\newcommand{\NP}{{\mathbb{NP}}}
\newcommand{\APX}{{\mathbb{APX}}}
\newcommand{\JRP}{{\mbox{JRP}}}
\newcommand{\JRPD}{{\mbox{JRP-D}}}
\newcommand{\JRPL}{{\mbox{JRP-L}}}

\newtheorem{theorem}{Theorem}

\newtheorem{lemma}{Lemma}

		{$\spadesuit$\smallskip}

\newenvironment{bigeqn*}{\large\begin{eqnarray*}}{\end{eqnarray*}}

\newcommand{\emdash}{\hspace{1mm}---\hspace{1mm}}

\newcommand{\OPT}{\textsc{Opt}}

\newcommand{\threehalves}{{\textstyle\frac{3}{2}}}
\newcommand{\onethird}{{\textstyle\frac{1}{3}}}

\newcommand{\twofifths}{{\textstyle\frac{2}{5}}}
\newcommand{\threefifths}{{\textstyle\frac{3}{5}}}

\newcommand{\nineeights}{{\textstyle\frac{9}{8}}}

\newcommand{\COST}{{\sf C}}

\newcommand{\cost}{{\sf c}}

\newcommand{\retailers}{\mathcal{R}}

\newcommand{\coststem}{\textsc{Cost}_\textsc{wship}}
\newcommand{\costbranch}{\textsc{Cost}_\textsc{rship}}
\newcommand{\costship}{\textsc{Cost}_\textsc{ship}}
\newcommand{\costwait}{\textsc{Cost}_\textsc{wait}}
\newcommand{\costtotal}{\textsc{Cost}}
\newcommand{\df}{D}
\newcommand{\I}{\mathcal{\xi}}
\newcommand{\shipment}{\mathrm{shp}}
\newcommand{\ALGOSRP}{1SRP}
\newcommand{\ALGTSRP}{2SRP}
\newcommand{\ALGLPS}{LPS}
\newcommand{\ALGCOMB}{1SRP+LPS}
\newcommand{\ALGFINAL}{2SRP+1SRP+LPS}
\newcommand{\shippace}{G}

\newcommand{\shippaceb}{G_\mathrm{1SRP}}
\newcommand{\shippacec}{G_\mathrm{LPS}}
\newcommand{\shippacebc}{G_\mathrm{1SRP+LPS}}

\newcommand{\dz}{\mathrm{d}z}
\newcommand{\firsttimelp}{\textsf{ft}_{{\mbox{\tiny LP}}}}
\newcommand{\scalefactor}{\zeta}

\newtheorem{observation}[theorem]{Observation}

\title{Better Approximation Bounds \\ for the Joint Replenishment Problem%
\thanks{Research partially supported by NSF grants CCF-1217314
and OISE-1157129, MNiSW grant no.\ N N206 368839, 2010-2013, EU ERC project 259515 PAAl, 
CE-ITI (project P202/12/G061 of GA~\v{C}R), and grant IAA100190902 of GA~AV~\v{C}R.}
}

\author[1]{Marcin Bienkowski}
\author[1]{Jaroslaw Byrka}
\author[2]{Marek Chrobak}
\author[1,3]{{\L}ukasz Je\.{z}}
\author[4]{Ji\v{r}\'{\i} Sgall}
\affil[1]{Institute of Computer Science, University of Wroc{\l}aw, Poland.}
\affil[2]{Department of Computer Science, University of California at Riverside, USA.}
\affil[3]{Department of Computer, Control, and Management Engineering, Sapienza University of Rome, Italy.}
\affil[4]{Computer Science Institute, Faculty of Mathematics and Physics, Charles University, Czech Republic.}

\date{}

\begin{document}

\maketitle

\begin{abstract}
The Joint Replenishment Problem ($\JRP$) deals with optimizing shipments of goods from a supplier to
retailers through a shared warehouse. Each shipment involves transporting goods from the
supplier to the warehouse, at a fixed cost $\COST$,
followed by a redistribution of these goods from the warehouse to the retailers
that ordered them, where transporting goods to a retailer $\rho$ has a fixed cost $\cost_\rho$.
In addition, retailers incur waiting costs for each order. The objective is to
minimize the overall cost of satisfying all orders, namely the sum of all
shipping and waiting costs.

$\JRP$ has been well studied in Operations Research and, more recently,
in the area of approximation algorithms. For arbitrary waiting cost functions, the best
known approximation ratio is $1.8$. This ratio can be reduced to
$\approx 1.574$ for the $\JRPD$ model, where there is no cost for waiting but orders have deadlines.
As for hardness results, it is known that the problem is $\APX$-hard and that the natural
linear program for $\JRP$ has integrality gap at least $1.245$. Both results hold even for $\JRPD$. 
In the online scenario, the best lower and upper bounds on the competitive ratio are
$2.64$ and $3$, respectively. The lower bound of $2.64$ applies even to the
restricted version of $\JRP$, denoted $\JRPL$, where the waiting cost function is
linear.

We provide several new approximation results for $\JRP$. In the offline case, we
give an algorithm with ratio $\approx 1.791$, breaking the barrier of $1.8$.
In the online case, we show a lower bound of $\approx 2.754$ on the competitive ratio
for $\JRPL$ (and thus $\JRP$ as well), improving the previous bound of $2.64$. 
We also study the online version of $\JRPD$, for which we prove that the optimal
competitive ratio is $2$.
\end{abstract}


\newpage


\section{Introduction}
\label{sec:introduction}

The Joint Replenishment Problem ($\JRP$) deals with optimizing shipments of goods from a supplier to a
set $\retailers$ of retailers through a shared warehouse. Over time, retailers issue orders for items.
All ordered items must be subsequently shipped to the retailers, although some shipments can be
delayed, in order to aggregate orders into fewer shipments to reduce cost.

Specifically, for each $\rho\in\retailers$ we are given
the cost $\cost_\rho$ of transporting goods from the warehouse to $\rho$. We are also given
the cost $\COST$ of transporting goods from the supplier to the warehouse. A shipment of goods
from the supplier to a subset $S\subseteq\retailers$ of retailers involves first shipping them to
the warehouse and then redistributing them to all retailers in $S$, at cost
equal $\COST + \sum_{\rho\in S}\cost_\rho$. Note that this cost is independent of the set of
items shipped. The waiting cost of an item $\pi$ ordered at time $a$ and delivered at time $t\ge a$ is
given by a function $h(t)$, possibly dependent on $\pi$,
where we assume that the values of $h(t)$ are non-decreasing with $t$.
The objective is to minimize the overall cost of satisfying all orders, namely the total cost
of shipments plus the total waiting cost of all orders.

There are two natural restrictions on waiting costs that have been previously considered in the
literature. One is to assume that the waiting costs are linear, that is
$h_{\pi}(t) = t-a_\pi$, where $a_\pi$ is the arrival time of an order $\pi$. 
We denote this version by $\JRPL$. In the other version, called
$\JRP$ with deadlines ($\JRPD$), there is no
waiting cost but ordered items must be shipped before pre-specified deadlines.

Several different, but mathematically equivalent, definitions of $\JRP$ can be found in the literature. 
In earlier papers $\JRP$ is phrased as an inventory management problem, where the inventory of
some commodity needs to meet a set of demands that arrive over time. The objective is to
balance the cost of orders\footnote{Note that the meaning of the term ``order'' here is different from our usage.}
that replenish the inventory with the cost of maintaining
it (the so-called holding cost). This formulation would not quite make sense in the
online scenario, since the orders that need to be scheduled take place before demands.
An online model of $\JRP$, referred to as \emph{Make-to-Order $\JRP$},
was introduced by Buchbinder~{\etal}~\cite{jrp-online-buchbinder}. In their description
there is no inventory; instead,
a collection of demands must be satisfied by subsequent orders. Except for minor
terminology variations,
our definition is essentially the same as that in \cite{jrp-online-buchbinder}.
Some of recent papers \cite{packet-aggregation-becchetti,aggregation-bkv,khanna-message-aggregation,online-control-wads}
on control message aggregation in networks,
introduce a model where control packets (corresponding to orders, in our definition) 
need to be transmitted to
a common destination (corresponding to the supplier), paying the transmission and delay costs. 
In particular, the flat-tree case studied in \cite{aggregation-bkv} is equivalent to $\JRPL$. 

$\JRP$ has been well studied in Operations Research and, more recently,
in the area of approximation algorithms. The problem is known to be strongly
$\NP$-hard, even for the special cases of $\JRPD$ and $\JRPL$
\cite{jrp-arkin,packet-aggregation-becchetti,jrp-deadlines-nonner}.
$\APX$-hardness proofs, even for some restricted versions of $\JRPD$, 
were given by Nonner and Souza~\cite{jrp-deadlines-nonner}
and Bienkowski~{\etal}~\cite{jrp-deadlines-icalp}.
The first approximation algorithm, with ratio $2$, was provided by
Levi, Roundy and Shmoys~\cite{jrp-levi-2-approx}, and was subsequently improved by
Levi et al.~\cite{jrp-owmr-levi-journal,jrp-owmr-levi-approx} to $1.8$
(see also \cite{jrp-owmr-levi-soda}).
For $\JRPD$, the ratio was reduced to $5/3$ by
Nonner and Souza~\cite{jrp-deadlines-nonner}
and then to $\approx 1.574$ by Bienkowski~\etal~\cite{jrp-deadlines-icalp}. All upper bounds are based on randomized
rounding of the natural linear program for $\JRP$.
As shown in \cite{jrp-deadlines-icalp}, the integrality gap of this linear program 
is at least $1.245$, even for $\JRPD$.

The online version of $\JRP$ was studied in the earlier discussed paper
by Buchbinder~{\etal}~\cite{jrp-online-buchbinder}, who give a 
$3$-competitive algorithm, using a primal-dual scheme, and show a lower bound of $2.64$
on the competitive ratio, even for $\JRPL$. (See also Brito~{\etal}~\cite{aggregation-bkv}
for related work.)


\myparagraph{Our contributions.}
We provide several new approximation results for $\JRP$. In the offline case, we
give an algorithm with approximation ratio $\approx 1.791$, breaking the barrier of $1.8$ from 
\cite{jrp-owmr-levi-journal,jrp-owmr-levi-approx}.
The improvement is achieved by refining the analysis of the LP-rounding algorithm in
\cite{jrp-owmr-levi-journal,jrp-owmr-levi-approx} and combining it with
a new algorithm that uses an approximation for $\JRPD$ from \cite{jrp-deadlines-icalp}.

We also study online algorithms for $\JRP$. We show that deterministic online algorithms,
even for $\JRPL$,
cannot be better than $\approx 2.754$-competitive, improving the bound of $2.64$
from~\cite{jrp-online-buchbinder}. 
For $\JRPD$, we prove that the optimal competitive ratio is $2$.

For convenience, we use a model where time is continuous, while some of previous
works on this topic used the discrete-time model. Algorithms 
for the continuous model can be easily translated into the discrete model,
preserving the same performance guarantee. In our lower bound proofs all
waiting-cost functions are left-continuous, and for such functions lower
bound arguments for competitive ratios carry over to the discrete case as well.
This relationship will be formally spelled out in the final version of this
paper (see a similar argument in \cite{jrp-online-buchbinder}).



\section{Preliminaries}\label{sec: preliminaries}

We now review our terminology and formalize the definition of $\JRP$.
Recall that $\retailers$ denotes the set of retailers. Each order can be specified by a triple
$\pi = (\rho,a,h)$, where  $a$ is the time when $\pi$ was issued, $\rho$ is the retailer that
issued $\pi$, and $h()$ is the waiting cost function of $\pi$, where $h(t) = \infty$ for $t <a$
and $h(t)$ is non-decreasing for $t\ge a$. Let $\Pi$ be the set of all orders.
In $\JRPL$ we will assume that
$h(t) = t-a$ for $t\ge a$, and in $\JRPD$ we have $h(t) = 0$ for $a \le t \le d$ and $h(t) = \infty$
otherwise. Then $d$ is called the \emph{deadline} of order $\pi$. In $\JRPD$ we will in
fact specify an order by a triple $\pi = (\rho,a,d)$.

A \emph{shipment} is specified by a pair $(S,t)$, where $S$ is the set of retailers receiving
the shipment and $t$ is the time of the shipment.
The cost of shipment $(S,t)$ is $\COST + \sum_{\rho\in S}\cost_\rho$.
A schedule is a set $\barS$ of shipments.
An order $\pi = (\rho,a,h)$ is said to be \emph{pending} in $\barS$ at time $\tau$ if
$a\le\tau$ and there is no shipment $(S,t)$ in $\barS$ with $\rho \in S$ and
$a\le t < \tau$.
If $\pi = (\rho,a,h)$ is pending at time $t$ and $(S,t)$ is a shipment in $\barS$
such that $\rho \in S$, then we say that $(S,t)$ \emph{satisfies} $\pi$.
In such case, the waiting cost of $\pi$ in $\barS$ is $h(t)$.
The cost of $\barS$ is the sum of its shipment and waiting costs, 
that is $\costtotal(\barS) = \costship(\barS) + \costwait(\barS)$, where
\begin{equation*}
	\costship(\barS) = \sum_{(S,t)\in \barS}(\COST + \sum_{\rho\in S}\cost_\rho)
			\quad\quad\textrm{and}\quad\quad
	\costwait(\barS) = \sum_{\pi =(\rho,a,h) \in \Pi}
					\min_{\substack{ (S,t)\in \barS \\\rho\in S,\,t\ge a}} h(t)
	\enspace,
\end{equation*}
where $\min\emptyset \equiv +\infty$.
The objective of $\JRP$ is to compute a schedule $\barS$ with minimum $\costtotal(\barS)$.


We use the standard definition of approximation algorithms.
We will say that a polynomial-time algorithm $\calA$ is an \emph{$R$-approximation algorithm}
for $\JRP$ if for any instance it computes a schedule of shipments 
whose cost is at most $R$ times the optimal cost for this instance.

In the online scenario, orders arrive
over time, and at each time $t$ an online algorithm must decide
whether to ship at time $t$ and, if so, to which retailers, based only
on the existing orders. For online algorithms we use the
term ``\emph{$R$-competitive}'' as a synonym of ``$R$-approximation''.

In the literature, some authors distinguish between absolute approximation
ratios (as defined above) and asymptotic ratios, where an algorithm is
allowed to pay some additional constant overhead cost, independent of the
instance. While our upper bounds apply to the absolute ratio, our
lower bound proofs can be extended to the asymptotic ratios by
repeating the lower bound strategies a sufficient number of times.


\section{An Upper Bound of 1.791 for Offline {\JRP}}

We now present our $1.791$-approximation algorithm.
The algorithm first computes an optimal solution $(x^*,y^*)$ of the
linear program for $\JRP$. Then it chooses randomly one of three 
different LP-rounding methods, with probabilities and other parameters suitably optimized,
to obtain a ratio improving the bound of $1.8$ from~\cite{jrp-owmr-levi-journal,jrp-owmr-levi-approx}.

\mymedparagraph{Linear program.}
Let $T = \{a : (\rho,a,h) \in \Pi\}$ be the times when orders are placed.
We can assume that all shipments occur at times in $T$. We use the following
indicator variables: $x_a$ represents a supplier-to-warehouse shipment at time $a$,
$x_{\rho,a}$ represents a warehouse-to-retailer $\rho$ shipment at time $a$, and
$y_{\pi,a}$ represents an order $\pi$ being satisfied by a shipment at time $a$.
The following linear program is the fractional relaxation of the 
natural integer program for JRP. 
\begin{alignat}{3}
\textrm{minimize}\hspace{1.2cm} 
    &  
        \textstyle \sum_{a \in T} \COST \cdot x_a\ 
        	&&+ \ \sum_{a \in T} \sum_{\rho\in\retailers} \cost_\rho \cdot x_{\rho,a}
        	+ \sum_{\pi =(\rho,a,h)\in \Pi} \sum_{t \in T: t \geq a} h(t) \cdot y_{\pi,t} 
    \notag \\
    \textrm{subject to}\quad\quad 
        \textstyle x_a  \;&\geq\;  x_{\rho,a} 
        & &\textrm{for all}\; a \in T, \rho \in \retailers  \\
    \textstyle x_{\rho,a} \;&\geq y_{\pi,a} \;
        & &\textrm{for all}\; \pi = (\rho,a,h) \in \Pi \\
    \label{eq:lp_y}
    \textstyle \sum_{t \geq a} y_{\pi,t} \;&\geq 1 \;        
        & &\textrm{for all}\; \pi = (\rho,a,h) \in \Pi \\
    \textstyle x_a, x_{\rho,a}, y_{\pi,a} \;&\ge 0\;
        & &\textrm{for all}\; a \in T, \rho \in \retailers,\pi\in\Pi
\end{alignat}
Throughout the rest of the paper, we will fix an optimal (fractional) solution
to the LP above and denote it by $(x^*,y^*)$. Note that 
constraints \eqref{eq:lp_y} are satisfied with equality in $(x^*,y^*)$.


\mymedparagraph{Algorithms~{\ALGTSRP} and~{\ALGOSRP}.}
The cost of any solution $(x,y)$ to the LP above can be  naturally split into
three parts: the supplier-to-warehouse shipping cost, $\coststem(x,y)$; the 
warehouse-to-retailers shipping cost, $\costbranch(x,y)$; and the waiting cost,
$\costwait(x,y)$.  When the solution $(x,y)$ is a random variable, these
denote appropriate  {\em expected} costs. 
We say that a solution $(x,y)$ is an \emph{ $(r_1,r_2,r_3)$-approximation} of $(x^*,y^*)$
if the following three conditions hold:
\begin{compactitem}
\item $\coststem(x,y) \leq r_1 \cdot \coststem(x^*,y^*)$, 
\item $\costbranch(x,y) \leq r_2 \cdot \costbranch(x^*,y^*)$, and 
\item $\costwait(x,y) \leq r_3 \cdot \costwait(x^*,y^*)$. 
\end{compactitem}
In our solution, we build on two LP-based, polynomial-time
algorithms of Levi~\etal~\cite{jrp-owmr-levi-journal}. Both are based on
random shifting.  The first one (denoted {\ALGTSRP}) is called 
\emph{Two-Sided Retailer Push Algorithm} the second one (denoted {\ALGOSRP}) is called 
\emph{One-Sided Retailer Push Algorithm}.

\begin{lemma}[\cite{jrp-owmr-levi-journal}]\label{lem:a1}
{\rm (a)}
Algorithm {\ALGTSRP} computes an integral solution $(x,y)$
that is a~$(1,2,2)$-approximation of the optimal fractional solution $(x^*,y^*)$.

{\rm (b)}
Algorithm {\ALGOSRP}, parameterized by $c \in (0,\half]$, computes an~integral solution $(x,y)$
that is a~$(\frac{1}{c},\frac{1}{1-c},\frac{1}{1-c}$)-approximation of the optimal fractional solution $(x^*,y^*)$.
\end{lemma}

The currently best known $1.8$-approximation algorithm~\cite{jrp-owmr-levi-journal}
is obtained by simply running {\ALGTSRP} with probability $\threefifths$ and 
{\ALGOSRP} with probability $\twofifths$, setting $c = \onethird$ in the latter.


\mymedparagraph{High-level idea.}
We start by showing that the $\costwait$ estimate of Algorithm~{\ALGOSRP}
in \lref[Lemma]{lem:a1}.b is not tight.  To analyze it more accurately, we
define a {\em shipping pace} of an~algorithm and show a connection between the shipping
pace and the waiting cost.  We use that to show that, for $c = \onethird$, 
Algorithm~{\ALGOSRP} computes in fact a $(3,\threehalves,\nineeights)$-approximation. This improvement alone does
not reduce the overall approximation ratio of the {\ALGOSRP}-and-{\ALGTSRP} combination, 
as it is still dominated by the retailer shipment cost ratio.

However, we will add a third ingredient to this combination: Algorithm~{\ALGLPS}, that
uses scaling of the fractional solution to obtain a new
fractional solution obeying certain deadlines and then applies the recent
result on $\JRPD$, the {\em deadline-constrained} variant of $\JRP$~\cite{jrp-deadlines-icalp}, 
to round it to an integral solution. By carefully choosing the scaling
factor, probabilities of choosing Algorithms~{\ALGTSRP}, {\ALGOSRP} and
{\ALGLPS}, and fine-tuning the choice of~$c$ in Algorithm~{\ALGOSRP},  we
eventually reduce the approximation ratio for JRP to about $1.791$.


\mymedparagraph{Shipping pace.}
To measure the waiting cost of an algorithm, we estimate how fast it
satisfies each particular order in comparison to how fast these orders are
satisfied in  $(x^*,y^*)$. In $(x^*,y^*)$, 
the orders can be thought of as being satisfied
gradually with time. In particular, a fraction $\sum_{t \in [a,t'] \cap T} y^*_{\pi,t}$ of 
an order $\pi = (\rho,a)$ is satisfied till time~$t'$ (inclusively). 
For any $\alpha \in [0,1)$, let
\begin{equation}
\label{eq:lp_time}
\textstyle
    \firsttimelp(\pi,\alpha) = \min \left\{ t \in T: t\ge a \;\textrm{and}\;
					\sum_{t' \geq t} y^*_{\pi,t'} \leq 1 - \alpha \right\}
    \enspace.
\end{equation}
In other words, $\firsttimelp(\pi,\alpha)$ is the first time when the yet un-satisfied
fraction of $\pi$ in $(x^*,y^*)$ is at most $1-\alpha$.

Let $\shippace : [0,1] \to \mathbb{R}_{\geq 0}$ be an integrable function such that 
$\int_0^1 \shippace(z)\, \dz = 1$. We say that a~(randomized) algorithm $\calA$ has a~{\em shipping pace}
$\shippace$ if for any order $\pi = (\rho,a) \in \Pi$ and $\alpha \in [0,1)$, it holds that 
\begin{equation}
\label{eq:pace}
    \Pr\Big[\,\textrm{$\calA$ ships at time $t \in [a, \firsttimelp(\pi,\alpha)]$}\,\Big] 
        \geq \int_0^\alpha \shippace (z) \, \mathrm{d}z
    \enspace.
\end{equation}
Note that a shipping pace is not unique; it is simply 
{\em a lower bound} on the shipping probability.


\begin{lemma}
\label{lem:pace_to_waiting_cost}
Let $\calA$ be a (randomized) algorithm with shipping pace $\shippace$ that produces a solution $(x,y)$. 
Then,
\begin{equation*}
    \costwait(x,y) \leq \costwait(x^*,y^*) \cdot 
        \sup_{w \in [0,1)} 
        \left\{ \frac{1}{1-w} \cdot \int_w^1 \shippace(z) \,\dz \right\}
        \enspace.
\end{equation*}
\end{lemma}

\begin{proof} 
We show that the relation above holds for the waiting cost of any individual order $\pi = (\rho,a,h)$.
For the sake of this proof, we number all the consecutive times from set $\{t \in T : t \geq a\}$ as 
$t_0 = a, t_1, t_2, \ldots$. Then the waiting cost associated with  $\pi$ is 
\begin{equation*}
	\textstyle
   \costwait^\pi(x,y) = \sum_{i \geq 0} h(t_i) \cdot \E[y_{\pi,t_i}] 
       = \sum_{i \geq 0} [ h(t_{i+1}) - h(t_i)] \cdot \sum_{j \geq i+1} \E[y_{\pi,t_j}]
  \enspace,
\end{equation*}
and the waiting cost of $\pi$ in $(x^*,y^*)$ can be expressed analogously (but
without taking expected values).
It is thus sufficient to compare $\sum_{j \geq i+1} \E[y_{\pi,t_j}]$ with  
$\sum_{j \geq i+1} y^*_{\pi,t_j}$. 
Let $\alpha = \sum_{j =0}^i y^*_{\pi,t_j} = 1- \sum_{j \geq i+1} y^*_{\pi,t_j}$.
Then, by the definition of \eqref{eq:lp_time}, 
$\firsttimelp(\pi,\alpha) \leq t_{i+1}$, and therefore
\begin{equation*}
\sum_{j \geq i+1} \E[y_{\pi,t_j}] 
= \Pr[\textrm{$\calA$ ships at time $t > t_{i+1}$}]
\leq \Pr[\textrm{$\calA$ ships at time $t > \firsttimelp(\pi,\alpha)$}]
\leq \int_\alpha^1 \shippace(z) \, \dz \enspace,
\end{equation*}
where we consider only shipments to $\rho$.
The equality holds because all $y_{\pi,t_j}$ are 0-1 variables and at most one is non-zero. 
The inequalities follow from $\firsttimelp(\pi,\alpha) \leq t_{i+1}$ and the 
definition of the shipping pace in \eqref{eq:pace}.
\end{proof}


\mymedparagraph{Waiting cost of Algorithm~{\ALGOSRP}}. 
We start with a brief description of Algorithm~{\ALGOSRP} (see Algorithm~2 in~\cite{jrp-owmr-levi-journal}).
The algorithm is parametrized by $c\in[0,\half]$.
It first computes the optimal fractional solution $(x^*,y^*)$ and then it schedules
the shipments, in two phases. In the first phase, it schedules the supplier-to-warehouse shipments.
Intuitively, one can visualize this schedule in terms of the
``virtual warehouse time'', equal to the accumulated fractional shipping value for the
warehouse, $X_t = \sum_{t'\le t}x_{t'}$. The algorithm chooses uniformly a random
$\psi\in[0,c]$ and schedules the shipments at virtual warehouse times 
$\psi, \psi+c,\psi+2c,...$, which then can be translated into real times. 
More formally, these shipments are scheduled at (real) times $t$ for which there is $i$ such 
that $X_{t-1} < \psi+ic \le X_t$.
In the second phase, we define tentative shipments from the warehouse to each retailer $\rho$.
This is done similarly, by choosing a random $\psi_\rho\in [0,1-c]$ and tentatively
scheduling these shipments at retailer $\rho$'s virtual times $\psi_\rho, \psi_\rho+1-c,\psi_\rho+2(1-c),...$.
For each tentative shipment of $\rho$, say at a (real) time $t$,
the actual shipment to $\rho$ will take place at the first time $t'\ge t$ for which
there is a supplier-to-warehouse shipment.


\begin{observation}\label{obs:1SRP-pace}
Algorithm~{\ALGOSRP}, with parameter $c \in (0,\half]$, has a shipping pace
\[
\shippaceb(z) = 
\frac{1}{1-c} \cdot
\begin{cases}
z / c & \mathrm{for}\;\;z \in [0,c), \\
1 & \mathrm{for}\;\;z \in [c,1-c), \\
(1-z) / c  & \mathrm{for}\;\;z \in [1-c,1].
\end{cases}
\]
\end{observation}

\begin{proof} 
(Sketch.)
Every order is analyzed as if it was waiting first for a shipment (in the computed
integral solution) at its
retailer and then at the warehouse, with the analysis carried out with respect to
the retailer's virtual time (the amount by which the fractional solution satisfies the order).
Then the waiting at the retailer has uniform
distribution $U[0, 1-c]$ and the waiting at the warehouse is upper bounded
with a uniform distribution $U[0,c]$. Hence, the distribution of the total
waiting time is bounded by a convolution of the two uniform distributions, which results
in the trapezoidal shape of the shipping pace $\shippaceb(x)$, see \lref[Figure]{fig:pace}.
\end{proof}

\noindent
{\em Side note.} We can use \lref[Observation]{obs:1SRP-pace} along with 
\lref[Lemma]{lem:pace_to_waiting_cost} to improve the waiting cost ratio 
of Algorithm~{\ALGOSRP}. Namely, the supremum of 
$\frac{1}{1-w} \cdot \int_w^1 \shippaceb(z) \,\dz$ is achieved for $z = c$ and is then 
equal to $(2-3c)/(2 (1-c)^2)$. Setting $c = \onethird$, we obtain that 
Algorithm~{\ALGOSRP} returns a $(3,\threehalves,\nineeights)$-approximation.
As we noted earlier, this result alone cannot improve the combination of 
Algorithms~{\ALGTSRP} and {\ALGOSRP}, because we improved only the third coefficient of the ratio.


\mymedparagraph{Algorithm~{\ALGLPS}.}
To improve the overall approximation guarantee, we therefore need to improve
the two first coefficients in the approximation ratio. To this end, we design a new
algorithm that performs well in terms  of
warehouse and retailer shipping costs and has a bounded waiting cost ratio.
This algorithm (see below) randomly scales up the optimum solution $(x^*,y^*)$,
then it converts the scaled solution into an instance of $\JRPD$, the variant of
$\JRP$ with deadlines, to which
it applies an approximation algorithm from \cite{jrp-deadlines-icalp}.
The algorithm uses a probability distribution $\df$ of the scaling parameter
that we will define later.

\medskip
\begin{center}
\begin{minipage}{6in}
\noindent
\hrulefill

\noindent
{\bf Algorithm}~{\ALGLPS} \\
\vspace*{-4ex}

\noindent
\hrulefill
\begin{compactenum}
\item Choose $\scalefactor \in (0,1]$ from a distribution with the
	density function  $\df : (0,1] \to \mathbb{R}_{\geq 0}$.
\item Compute an optimal fractional solution $(x^*,y^*)$.
\item Create a new fractional solution $(\widehat{x}, \widehat{y})$ by setting $\widehat{x} = 
    \min\{1, x^*/\scalefactor \}$, and (greedily) choosing $\widehat{y}$ to minimize the waiting cost,
    subject to fixed fractional shipments $\widehat{x}$.  
\item Create an instance $\mathcal{L}$ of $\JRPD$, by inserting a deadline for each order 
$\pi$ at the first time $t'$ for which $\sum_{t\le t'} \hat{y}_{\pi,t} \geq 1$.
(Thus in $(\widehat{x}, \widehat{y})$ each order is served ``just in time''.)
\item Solve instance $\mathcal{L}$ by using the $\lambda$-approximation algorithm 
from~\cite{jrp-deadlines-icalp}, where $\lambda \approx 1.574$, and return the obtained solution.
\end{compactenum}
\vspace{-2ex}
\noindent
\hrulefill
\end{minipage}
\end{center}


\begin{lemma}
\label{lem:lps_costs}
Let $\I = \int_0^1 \frac{1}{z} \cdot \df(z) \,\mathrm{d}z$.
Algorithm {\ALGLPS} produces an integral solution $(x,y)$ with
\begin{align*}
\coststem(x,y) &\leq \lambda \cdot \I \cdot \coststem(x^*,y^*) \quad \textrm{and}
\\
\costbranch(x,y) &\leq \lambda \cdot \I \cdot \costbranch(x^*,y^*). 
\end{align*}
\end{lemma}

\begin{proof}
We analyze the output $(x,y)$ of Algorithm~{\ALGLPS} for a fixed $\scalefactor \in (0,1]$.
By Step~3, $\coststem(\widehat{x}, \widehat{y}) \leq (1/\scalefactor) \cdot \coststem(x^*,y^*)$
and by Step~5, $\coststem(x,y) \leq \lambda \cdot \coststem(\widehat{x},\widehat{y})$.
Thus, $\coststem(x,y) \leq (\lambda / \scalefactor) \cdot \coststem(x^*,y^*)$. 
By integrating the estimate above 
over the probability distribution of $\scalefactor$, we immediately obtain 
the first property of the lemma. 
The proof for the second property is analogous.
\end{proof}

\begin{lemma}
\label{lem:lps_shipping}
Fix any $\scalefactor \in (0,1]$ and let $(x,y)$ be the solution returned by 
Algorithm~{\ALGLPS} for this fixed $\scalefactor$.
Fix also an order $\pi = (\rho,a,h) \in \Pi$. Let $\shipment(\pi) \geq a$ be the time 
of the shipment that satisfies $\pi$ in $(x,y)$, that is
$y_{\pi,\shipment(\pi)} = 1$ and $y_{\pi,t} = 0$ for $t\neq\shipment(\pi)$.
Then  $\sum_{t \geq \shipment(\pi)} y^*_{\pi,t} \geq 1 - \scalefactor$.
\end{lemma}

\begin{proof}
The solution $(x,y)$ obeys the deadlines of instance $\mathcal{L}$.
The deadlines are set exactly
to satisfy the bound of the lemma, i.e., for each order $\pi$ at least
$1 - \scalefactor$ fraction of $\pi$ is still to be sent in $(x^*,y^*)$ at the
time of the deadline inserted for $\pi$.
\end{proof}


\begin{observation}
\label{obs:lps_pace}
Algorithm~{\ALGLPS} has a shipping pace $\shippacec \equiv \df$.
\end{observation}

\begin{proof}
By \lref[Lemma]{lem:lps_shipping}, for any $\pi$ and $\scalefactor$,
Algorithm~{\ALGLPS} plans a shipment for $\pi$ no later than when
a fraction $\scalefactor$ of $\pi$ is satisfied by $(x^*,y^*)$.
Thus $\df$ is indeed a shipping pace of the algorithm.
\end{proof}


\mymedparagraph{Combining Algorithms~{\ALGOSRP} and {\ALGLPS}.}
Algorithm~{\ALGCOMB} simply runs  Algorithm~{\ALGOSRP}
with probability $p$ and Algorithm~{\ALGLPS} with probability $1-p$.  (Recall that we
still have to choose parameter $c$ in Algorithm~{\ALGOSRP} and the
probability  density $\df$ in Algorithm~{\ALGLPS}.) We observe that such
an algorithm has pace $\shippacebc \equiv p \cdot \shippaceb + (1-p) \cdot \shippacec$. The
following result is an immediate consequence of (i) using
\lref[Lemma]{lem:a1}.b and \lref[Lemma]{lem:lps_costs} to estimate the total
warehouse and retailer shipping costs, (ii) applying
\lref[Lemma]{lem:pace_to_waiting_cost} 
and \lref[Observation]{obs:lps_pace} 
to estimate the waiting costs, and
(iii) using the inequality  $\sup_{w \in [0,1]} (\int_w^1 \shippace(z) \, \dz) / (1-w)
\leq \sup_{z \in [0,1]} \shippace(z)$.


\begin{lemma}
\label{lem:alg_BC}
Let $\I = \int_0^1 \frac{1}{z} \cdot \df(z) \,\mathrm{d}z$. 
Algorithm {\ALGCOMB} computes an integral solution $(x,y)$ that is a $(r_1,r_2,r_3)$-approximation 
of the optimal fractional solution $(x^*,y^*)$, where
\begin{equation*}
	r_1 = p/c + (1-p) \lambda \I,
	\quad\quad
	r_2 = p/(1-c) + (1-p) \lambda \I,
	\quad
	\textrm{and}
	\quad\quad
	r_3 = \sup_{z \in [0,1]} \shippacebc(z).
\end{equation*}
\end{lemma}

\begin{figure}[t]
\begin{center}
\vspace{-0.35in}
\includegraphics[width=0.75\textwidth]{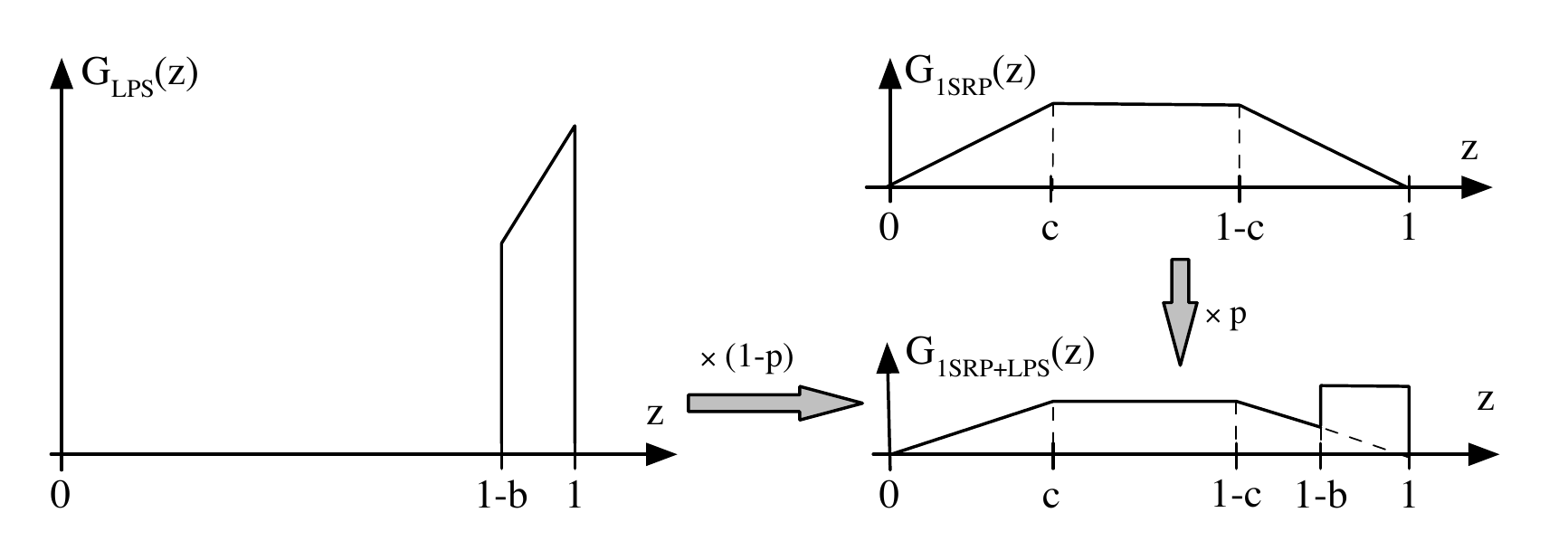}
\vspace{-0.35in}
\end{center}
\caption{Shipping paces of Algorithms~{\ALGOSRP}, {\ALGLPS} and {\ALGCOMB}.}
\label{fig:pace}
\vspace{-0.07in}
\end{figure}


\medskip

Our next step is to choose a probability density function $\df$ in Algorithm~{\ALGLPS}. 
This choice affects the approximation ratio in two ways.
On the one hand, we want the value of $\I$ in~\lref[Lemma]{lem:alg_BC}
to be as small as possible.
To this end, the probability mass should be accumulated close point $1$.
On the other hand, we need to take into account that
Algorithm~{\ALGCOMB} approximates the waiting cost within the factor 
$\sup_{z \in [0,1]} \shippacebc(z)$,
where $\shippacebc \equiv p \cdot \shippaceb + (1-p) \cdot \shippacec = 
 p \cdot \shippaceb + (1-p) \cdot \df$.

Therefore, we choose $\df$ to be supported on the interval $[1-b,1]$,
where $b \leq c$ is a parameter that we will fix later.
Furthermore, we choose $\df$ to be such an increasing linear function
that the resulting function $\shippacebc$ is constant on $[1-b,1]$, 
cf.~\lref[Figure]{fig:pace}.
To this end, we require that{\emdash}within the interval $[1-b,1]${\emdash}the 
slope of the function $(1-p) \cdot \shippacec \equiv (1-p) \cdot \df$ 
matches the negated slope of the function $p \cdot \shippaceb$.
These considerations imply that once we fix parameters $p$, $c$ and $b$, the probability density $\df$ 
should be 
\begin{equation}
\df(z) = \alpha \cdot z + \frac{1}{b} + \frac{\alpha \cdot b}{2} - \alpha \enspace,
\quad\textrm{where} \quad
\alpha = \frac{p}{(1-p) \cdot c \cdot (1-c)}
\enspace,
\end{equation}
for $z \in [1-b,1]$ and zero outside of this interval.
Straightforward calculations verify that $\df$ is indeed a probability density, i.e., 
$\int_0^1 \df(x) \, \mathrm{d}x = 1$.
Furthermore,  $\shippacebc \equiv p \cdot \shippaceb + (1-p) \cdot \shippacec$
is constant on the interval $[1-b,1]$ and its value there is equal to 
$\shippacebc(1) = (1-p) \cdot \df(1) = (1-p)/b + p b / (2 c (1-c))$.
Thus
\begin{equation*}
\I \; 
    =  \; \int_0^1 \frac{1}{z} \cdot \df(z) \,\mathrm{d}z \\
    =  \; \int_{1-b}^1 \alpha \,\mathrm{d}z + 
        \left(\frac{1}{b} + \frac{\alpha \cdot b}{2} - \alpha \right) \int_{1-b}^1 \frac{1}{z} \,\mathrm{d}z \\
    = \; \alpha \cdot b - \left(\frac{1}{b} + \frac{\alpha \cdot b}{2} - \alpha \right) \cdot \ln (1-b)
    \enspace.
\end{equation*}

We numerically optimize the parameters $p$, $c$ and $b$. Specifically, we choose 
$p=0.822599$, $c=0.342538$ and $b=0.136366$. For those 
values the maximum of function $\shippacebc$ is achieved in the interval $[1-b,1]$. 
and is at most $1.549968$.
Using the bounds of \lref[Lemma]{lem:alg_BC}, we conclude that 
Algorithm~{\ALGCOMB} is an
$(R_1,R_2,R_2)$-approximation of $(x^*,y^*)$, where $R_1 \leq 2.700277$ and $R_2 \leq 1.549968$. 


\mymedparagraph{Making ends meet: combining all algorithms.}
Finally, we combine algorithm {\ALGCOMB} with algorithm {\ALGTSRP}. The resulting algorithm {\ALGFINAL}
uses {\ALGTSRP} with probability $(R_1-R_2)/(R_1-R_2+1)$ and {\ALGCOMB} with probability $1/(R_1-R_2+1)$. 
Such algorithm is a
$(R,R,R)$-approximation (of the fractional optimal solution), where 
\[
    R = \frac{2 \cdot R_1 - R_2}{R_1 - R_2 + 1} \leq 1.790713
    \enspace.
\]
We therefore obtained the following result. 

\begin{theorem}
There is a polynomial-time $1.791$-approximation algorithm for $\JRP$.
\end{theorem}


\section{A Lower Bound of 2.754 for Online {\JRPL}}
\label{sec: lower bound 2.754 for online jrlp}

We now show our lower bound of $2.754$ for the competitive ratios for $\JRP$,
which improves the previous upper bound of $2.64$ by
Buchbinder~{\etal}~\cite{jrp-online-buchbinder}. Since we use only linear
waiting-cost functions in our construction, as in~\cite{jrp-online-buchbinder},
our result applies to $\JRPL$ as well.


\myparagraph{Single-phase game.} In our lower-bound proof it will be convenient to 
consider a simple version of $\JRPL$ that we refer to as the \emph{Single-Phase $\JRPL$}.
In the Single-Phase $\JRPL$ all orders are issued at the beginning at time $0$.
The waiting cost is assumed to be linear.
In addition to the set of retailers and orders, the instance specifies also
an \emph{expiration time} $\theta$. At time $\theta$ all orders expire:
they need not be satisfied anymore, but each incurs the waiting cost $h(\theta)=\theta$.
Note that all information about the instance is known to the online algorithm,
except for $\theta$, which represents the adversary strategy.
Thus the Single-Phase $\JRPL$
is in fact a generalization of the well studied rent-or-buy problem.

We claim that a lower bound of $R$ for Single-Phase $\JRPL$ implies a lower
bound of $R$ for $\JRPL$ (and thus for $\JRP$ as well). 
Since a similar argument appeared before 
in \cite{jrp-online-buchbinder,online-control-wads}, we only briefly sketch the proof of this claim.
Suppose that we have an adversary strategy that forces ratio $R$ for
Single-Phase $\JRPL$. We modify it into an adversary strategy that
forces the same ratio for $\JRPL$. This strategy creates a large number of
single-phase instances, concatenated together, with the $i$-th
instance scaled by a factor of $M^i$, for some very large $M$, in the
following sense: each order is replaced by $M^i$ identical orders and the
time is accelerated by a factor of $M^i$ as well. By accelerating the time
we mean that all time values used to make decisions in the strategy
are multiplied by $M^{-i}$. The adversary applies the same strategy in 
each phase, forcing ratio $R$ for each phase. Each phase may produce
some number of non-satisfied orders, but these can be satisfied by
one shipment for all retailers at the end of the game.
This will add only a constant to the adversary shipment cost and, since
the phase lengths are decreasing so fast, the increase of the 
adversary's waiting cost will be also negligible.


\myparagraph{Single-phase construction.}
We will use an instance of Single-Phase~{\JRPL} with
$N+1$ retailers in $\retailers$, denoted $\rho_0,\rho_1,\ldots,\rho_N$.
The costs of shipping from the warehouse to each of them are as follows:
$\cost_{\rho_0}=\cost_0=0$ and $\cost_{\rho_i} = \cost$ for all $i>0$
and $\cost$ that we fix later.  These are normalized so that $\COST=1$,
i.e., the cost of shipping from the supplier to the warehouse is $1$.

Each retailer $\rho_i$ places $w_i$ identical orders $(\rho_i,0,h_i)$ at time $0$,
where $h_i(t) = t$ for all $i$. Equivalently, we can view this as issuing a
single order $\pi_i=(\rho_i,0,h'_i)$ with weight $w_i$, that is with 
waiting cost function  $h'_i(t) = w_i\cdot t$. We will
adapt this terminology in this section. We choose the weights to be quickly decreasing, that is
$w_i \gg w_{i+1}$ for all $i < N$, so that the slopes of the functions $h'_i$ are decreasing
rapidly with $i$. As a result, in the proof below, when the algorithm satisfies
an order $\pi_i$, the waiting costs (of the algorithm and the adversary)
of all orders $\pi_{i+1},\pi_{i+2},...$ will
be negligible. For clarity, in the calculations below we will assume these costs to be $0$. 
(By adjusting the weights appropriately, we can make these costs at most an 
arbitrarily small $\epsilon$, and then our lower bound will approach $R$.)

Let  $\calA$ be an online algorithm for Single-Phase~{\JRPL}.
To describe the adversary strategy, we first normalize the way $\calA$ proceeds.
Using a simple exchange argument, it is
easy to show that, without loss of generality,
$\calA$ satisfies all the demands in increasing order of their
indices, i.e., if $i < j$ then $\pi_i$ is satisfied earlier than or together with $\pi_j$.
Then, the adversary stops the game the moment that $\calA$ satisfies more than one
order with a single shipment. (``Stopping'' means that the expiration time $\theta$
is set to the current time.) To complete the strategy's description we can
thus focus on $\calA$ satisfying orders $\pi_0,\pi_1,\ldots$ in this order,
each with a~dedicated shipment.  If the waiting cost associated with $\pi_i$
at the moment of its satisfaction is smaller than a~certain threshold value
$\sigma_i$, the game ends, otherwise it continues.  This means that as long
as the game did not end at of before $\calA$'s shipment satisfying $\pi_i$,
$\calA$'s cost for these shipments is at least $\sum_{j=0}^{i} (1+\cost_j + \sigma_j)$.
In particular, if the game does not end due to any aforementioned reason at
or before the time that $\calA$ satisfies $\pi_N$, then the game ends
\emph{naturally} with this shipment; otherwise we say that the game ends
\emph{prematurely}.

As was the case with $\cost_i$'s, all the thresholds $\sigma_i$ coincide and
are denoted $\sigma$, with the exception of $\sigma_0$.  We now give the values
of all the parameters.  We let $\cost$ be the only real root of
\begin{equation}\label{eq:c-equation}
 \cost^2(\cost+1)=1 \enspace,
\end{equation}
%
and
\begin{equation}\label{eq: sigma_0 and sigma}
 \sigma_0 \equiv \frac{1}{\cost+1} = \cost^2 \enspace,
	\quad\quad\quad
 \sigma \equiv \sigma_0^2 = \cost^4 = \cost^2 + \cost -1 \enspace,
\end{equation}
where the identities follow from~\eqref{eq:c-equation}.
We have $\cost\approx 0.7548$, 
$\sigma_0 \approx 0.5698$ and $\sigma\approx 0.3247$.

We claim that unless the game ends naturally, the competitive ratio of $\calA$
is at least $R=2+\cost$.  To see this, let us consider all the ways
in which the game can end prematurely.

Let $\omega$ be $\calA$'s waiting cost of $\pi_0$ when it satisfies $\pi_0$. If $\omega < \sigma_0$
then $\calA$'s cost is at least $1+\omega$, whereas {\OPT} can pay the waiting cost $\omega$ alone,
resulting in ratio no smaller than
\[
 1+\frac{1}{\sigma_0} = 2+\cost = R \enspace.
\]

If $\calA$ satisfies $\pi_0$ together with another order
by a single shipment, then $\calA$'s cost is at least $1+\cost+\omega$ whereas {\OPT} will either
pay the waiting cost $\omega$ for $\pi_0$ or $1$ for satisfying $\pi_0$ at time $0$.  Thus
the competitive ratio is at least
\[
 \frac{1+\cost+\omega}{\min\{1,\omega\}} \geq \frac{2+\cost}{1} = 2+\cost = R \enspace.
\]

Now we consider analogous two cases regarding the shipment for $\pi_i$, where $i\ge 1$,
assuming that the game did not end before. This means that $\calA$ already suffered a~cost of at least
$\sigma_0+1 +(i-1)(\sigma+1+\cost)$ 
associated with satisfying orders $\pi_0,...,\pi_{i-1}$,
plus some additional cost associated with satisfying $\pi_i$.
Let now $\omega$ denote the waiting cost of $\pi_i$ when $\calA$ satisfies $\pi_i$.

If $\omega < \sigma$
then {\OPT} satisfies the orders $\pi_j$ for all $j<i$ with a single shipment at time $0$,
and pays the waiting cost $\omega$ for $\pi_i$.  The competitive ratio is at least
\begin{equation*}
 \frac{ \sigma_0+1 +(i-1)(\sigma+1+\cost) + \omega+1+\cost}{1+(i-1)\cost+\omega}
 = 1+\frac{i+\cost+\sigma_0+(i-1)\sigma}{1+(i-1)\cost+\omega}
 \geq 1+\frac{i+\cost+\sigma_0+(i-1)\sigma}{1+(i-1)\cost+\sigma} \enspace,
\end{equation*}
which after substituting formulas~\eqref{eq: sigma_0 and sigma} for $\sigma_0$ and $\sigma$,
as well as using~\eqref{eq:c-equation}, becomes
\begin{equation*}
 1 + \frac{1+i\cost+i\cost^2}{i\cost+\cost^2} = 2 + \frac{1+i\cost^2 - \cost^2}{\cost(i+\cost)}
 = 2+ \frac{i\cost^2+\cost^3}{\cost(i+\cost)} = 2 + \cost = R \enspace.
\end{equation*}

Let us consider the remaining case in which $\calA$ satisfies another order together with $\pi_i$.
In this case {\OPT} satisfies all the previous orders with a~single shipment at time $0$; as for
$\pi_i$, {\OPT} either satisfies it with that shipment as well, or pays the waiting cost $\omega$ for $\pi_i$,
whichever is cheaper.  Thus the ratio is at least
\begin{align*}
 \frac{ \sigma_0+1 +(i-1)(\sigma+1+\cost) + \omega +1+2\cost}{1+(i-1)\cost+\min\{\cost,\omega\}}
 		&\geq 1+ \frac{i+2\cost+\sigma_0 + (i-1)\sigma}{1+i\cost}\enspace,
\end{align*}
which after substituting formulas~\eqref{eq: sigma_0 and sigma} for $\sigma_0$ and $\sigma$, becomes
\begin{equation*}
 1+ \frac{1+(i+1)\cost+i\cost^2}{1+i\cost} = 1+ \frac{(1+\cost)(1+i\cost)}{1+i\cost} = 2+\cost = R \enspace.
\end{equation*}

Thus the ratio is at least $R$ if the game ends prematurely.  But if it does not, then $\calA$'s cost for
each shipment, except the one for $\pi_0$, is at least $1+\cost+\sigma=\cost^2+2\cost$, 
by \eqref{eq: sigma_0 and sigma}.
On the other hand, {\OPT} satisfies all orders with a~single shipment at time $0$, which costs $1+N\cost$.
With $N\to\infty$, {\OPT}'s cost of $1$ for shipment from the supplier to the warehouse becomes
negligible and {\OPT}'s cost per order tends to $\cost$.  Therefore, the competitive ratio tends to $R=2+\cost$.
Summarizing the above argument, we obtain:

\begin{theorem}\label{thm: online jrp lower bounds}
Each online deterministic algorithm for $\JRPL$ has competitive ratio at least $2.754$.
\end{theorem}


\section{Tight Bound of 2 for Online {\JRPD}}
\label{sec: tight bound of 2 for online jrpd}


We now present an online algorithm for $\JRPD$ with competitive ratio $2$, matching the
lower bound that is given in \lref[Appendix]{sec: a lower bound 2 for online jrpd}.
We will denote the shipments of the algorithm by
$(B_1,t_1)$, $(B_2,t_2)$,..., where $t_1\le t_2\le ...$.
The set $B_j$ of retailers participating in the $j$-th shipment
is called the $j$th \emph{batch}. For convenience, we introduce
a ``dummy'' $0$'th shipment at time $t_0 = 0$, which we think of as if it
shipped to all the retailers in the instance at no cost.
(All that matters is that at time $0$ each retailer does not have any pending orders.)
For a retailer $\rho$ and time $t$, we define the \emph{deadline} of $\rho$ to be the
earliest deadline of a pending order in $\rho$. If a retailer does not have any pending orders,
its deadline is $+\infty$.
Without loss of generality, we can assume that
all shipments (of an online algorithm and the adversary)
take place only at deadlines of some retailers.
If $t_j$ is the deadline of a retailer $\rho$ then we say that $\rho$, or the order in $\rho$ with
deadline $t_j$, \emph{triggers} shipment $(B_j,t_j)$. 


\myparagraph{Algorithm $\calG$:}
Suppose that we just completed shipment $(B_{j-1},t_{j-1})$.
We wait until we reach a deadline of a retailer, which will become the trigger retailer for the $j$th shipment.
We denote this retailer by $\chi_{j}$ 
and its deadline by $t_{j}$. At time $t_{j}$ our batch is
$B_{j} = \braced{\chi_{j}}\cup X_{j}$, where $X_{j}$ contains the maximum number
of retailers, in order of increasing deadlines, such that $\cost(X_{j})\le \COST$.

\medskip

If $j=1$ then, according to our convention, $j-1 = 0$ refers to the dummy
shipment at time $t_0 = 0$. Thus the first shipment will occur at the
first deadline of the instance.


\myparagraph{Analysis.}
We now analyze this algorithm. To simplify the analysis we will assume that
all order arrival times and deadlines are different.
The instance can be converted to have this property by an infinitesimal 
perturbation of arrival times and deadlines.

We divide the sequence of shipments into \emph{phases}. A phase is a maximal
interval $[g,h]$ of integers (indices of shipments), where $1 \le g \le h$, 
such that the adversary does not make any
shipments in the time interval $(t_g,t_h]$. In other words:
{(i)} there are no adversary shipments in $(t_g,t_h]$,
{(ii)} the adversary shipped in $(t_{g-1},t_g]$, and
{(iii)} either $t_h$ is the last deadline or the adversary shipped in $(t_h,t_{h+1}]$.
Note that the first phase starts with the first shipment (that is $g=1$). Indeed,
the adversary must ship in the interval $(t_0,t_1]$, because $t_1$ is the first
deadline. The lemma below elucidates a property of
phases that will be critical to our analysis.


\begin{lemma}\label{lem: jrdp phase structure}
Let $[g,h]$ be a phase and $g < j \le h$. Let $\pi$ be the order in $\chi_j$ that triggers
shipment $B_j$. Then $\pi$ was pending at time $t_{j-1}$, and among all orders
pending at time $t_{j-1}$ it was the earliest-deadline order not in $B_{j-1}$.
\end{lemma}

\begin{proof}
Suppose that $\pi  = (\chi_j,a,d)$, that is $d = t_j$. 
If we had $a > t_{j-1}$ then the adversary
would have to make a shipment in the interval $(a,d]$, but this would contradict the
definition of a phase. (Recall that, by our assumptions, the
adversary cannot ship at time $a$.) So $\pi$ was pending at time $t_{j-1}$.
It must also be in fact the earliest-deadline order outside $B_{j-1}$, because
$B_j$ is the first shipment after $B_{j-1}$.
\end{proof}


\begin{lemma}\label{lem: jrdp core shipments}
Let $[g,h]$ be a phase and $g < j \le h$.
Suppose that $\rho \in B_j$, where for $j=h$ we assume that $\rho = \chi_h$.
Let $j'< j$ be maximum such that $\rho\in B_{j'}$ (if there is
no shipment to $\rho$ before $t_j$, let $j' = 0$). 
Then the adversary must ship to $\rho$ in the interval
$(t_{j'},t_h]$, whence $j' < g$.
\end{lemma}

\begin{proof}
Right after the
shipment at time $t_{j'}$, there were no orders in $\rho$, so there
must be an order that arrived after $t_{j'}$ and has deadline
at most $t_h$, by the algorithm, and by the fact that $\rho\notin X_h$.
So the
adversary must ship in the time interval $(t_{j'},t_h]$.	This in turn
implies that $j' < g$, by the definition of phases.
\end{proof}

Consider a phase $[g,h]$.
Using the above lemma, if $\rho\in B_j$, where either $g\le j < h$ or $j=h$ and
$\rho = \chi_h$, then with the $\calG$'s warehouse-to-$\rho$ shipment at time $t_j$
we can associate a unique warehouse-to-$\rho$ shipment of the adversary 
that occurred not later than at time $t_h$.
With this in mind, we can now analyze Algorithm~$\calG$ using a charging
argument, as follows:
\begin{compactitem}
\item We charge $\cost(\chi_g)$, namely the cost of the warehouse-to-$\chi_g$ shipment
	at time $t_g$ to the associated warehouse-to-$\chi_g$ shipment of the adversary
	(as described above). The charging ratio here is $1$.
\item We charge $\COST + \cost(X_h)$, representing the cost of the first supplier-to-warehouse
shipment at time $t_g$ and the cost of the warehouse-to-$X_h$ shipment at time $t_h$, to
the adversary cost of $\COST$ of the supplier-to-warehouse shipment cost right before $t_g$.
Since $\cost(X_h) \le \COST$, the charging ratio is at most $2$.
\item For $j = g,...,h-1$, we charge the cost
$\COST + \cost(X_j) + \cost(\chi_{j+1})$, that represents the supplier-to-warehouse shipment cost
at time $t_{j+1}$ and the cost of shipments warehouse-to-$X_j$
and warehouse-to-$\chi_{j+1}$, to
$\cost(X_j) + \cost(\chi_{j+1})$, namely the adversary's warehouse-to-retailer shipment cost associated with
the retailers in $X_j\cup\braced{\chi_{j+1}}$. 
By the choice of $X_j$ and \lref[Lemma]{lem: jrdp phase structure},
we have $\cost(X_j) + \cost(\chi_{j+1}) > \COST$, so the charging ratio is at most $2$.
\end{compactitem}
In all cases the charging ratio is at most $2$, and different charges are
assigned to different portions of the adversary cost. Thus
Algorithm~$\calG$ is $2$-competitive. In \lref[Appendix]{sec: a lower bound 2 for online jrpd}
we show a matching lower bound. 

\begin{theorem}\label{thm: G is 2-competitive}
{\rm (a)}
Algorithm~$\calG$ is $2$-competitive for $\JRPD$.
{(\rm b)} Every deterministic online algorithm for $\JRPD$
has competitive ratio at least $2$.
\end{theorem}


\section{Final Comments}

There are still significant gaps between the lower and upper bounds for the approximability of different
variants of $\JRP$. For $\JRPD$, we have $\APX$-hardness 
\cite{jrp-deadlines-nonner,jrp-deadlines-icalp}
and an integrality gap of $1.245$ \cite{jrp-deadlines-icalp}, while the
best upper bound is $1.574$ \cite{jrp-deadlines-icalp}. 
The case of $\JRPL$ (linear waiting costs), although most natural, is poorly understood. The
best upper bound is the same as for the general case, namely $1.791$, shown in this paper,
even though at this time not even an approximation scheme has been ruled out.
Some progress on this problem was recently reported in~\cite{jrp-scheme-nonner-ipco13}.
The approximability of the online version of $\JRPL$ also remains open.

{\JRP} can be naturally generalized to trees of arbitrary depth. This multi-level $\JRP$
problem with deadlines was studied by Bechetti~{\etal}~\cite{packet-aggregation-becchetti}, 
who provided a $2$-approximation algorithm. Khanna~{\etal}~\cite{khanna-message-aggregation}
considered the case of linear waiting costs.
Very recently, Chaves (private communication) has shown that
the general case can be reduced to the so-called
multi-stage assembly problem, for which a $2$-approximation algorithm was
given by Levy~{\etal}~\cite{jrp-levi-2-approx}.


\newpage


\bibliographystyle{plain}
\bibliography{references}

\appendix


\section{A Lower Bound of 2 for Online {\JRPD}}
\label{sec: a lower bound 2 for online jrpd}


In this section we show that no online algorithm for $\JRPD$ can have competitive ratio smaller than $2$,
thus proving \lref[Theorem]{thm: G is 2-competitive}(b).
Similar to the proof in \lref[Section]{sec: lower bound 2.754 for online jrlp}, 
we actually provide this lower bound for the restricted
variant of $\JRPD$ called \emph{Single-Phase $\JRPD$}.
In Single-Phase $\JRPD$, all orders arrive at time $0$. The adversary
can stop the game at any time $\theta$ (the expiration time), unknown to the online algorithm. 
All orders not satisfied by time $\theta$ incur no cost.
By an argument similar to the one given in \lref[Section]{sec: lower bound 2.754 for online jrlp},
any lower bound for Single-Phase $\JRPD$ implies the same lower bound for $\JRPD$.
(For $\JRPD$ the argument in \lref[Section]{sec: lower bound 2.754 for online jrlp} 
has to be slightly refined; the details will be given in the final version of this paper.)

In our instance, the supplier-to-warehouse shipping cost is $\COST = 1$. 
We have $N+1$ retailers $\rho_i$, $i=0,1,...,N$, for some sufficiently large $N$.
Retailer $\rho_0$ has shipping cost $\cost_{\rho_0} = 0$, and each
retailer $\rho_i$, for $i>0$, has shipping cost $\cost_{\rho_i} = 1$.
For each $i$, retailer $\rho_i$ issues one order $\pi_i$ at time $0$ with
waiting cost function
\begin{equation*}
	h_i(t) = \left\{\begin{array}{lcl}
					0 &\quad& 0 \le t \leq i
					\\
					\infty && 	t > i
				\end{array}
				\right.
\end{equation*}
Let $\calA$ be an online algorithm for $\JRPD$.
$\calA$ must ship to each $\rho_i$ no later than at time $i$.
Without loss of generality, $\calA$ ships only at integer times, so
as long as $\calA$ ships to each retailer separately then each $\rho_i$
will be shipped at time $i$. 
The adversary will stop the game as soon as $\calA$ ships to more than one
retailer. If this does not happen, the game stops after the shipments to
all retailers, that is right after time $N$.

We argue now that this forces the competitive ratio of $\calA$ to be
arbitrarily close to $2$.
If $\calA$ ships to each retailer separately, it pays $1$ for retailer $0$ and
$2$ for each retailer $\rho_i$, $i\ge 1$, for the total cost of $2N+1$. 
The adversary can ship to all retailers at the
beginning, paying $N+1$. So the ratio approaches $2$ with $N\to\infty$.

Suppose that at some time $k$, $\calA$ ships to $\rho_k$ and some other retailer, and
let $k$ be the first such $k$.
Then $\calA$'s cost is $1 + 2(k-1) + 3 = 2k+2$. The adversary can ship to
all retailers $\rho_0,...,\rho_k$ at the beginning, paying $k+1$. So the
ratio is $2$.

\end{document}